\documentclass[a4paper,11pt]{easychair}
\usepackage{charter}
\usepackage{amsfonts}
\usepackage{amsthm}
\usepackage{amssymb}
\usepackage{amsmath}
\usepackage{xspace}
\usepackage{fullpage}
\usepackage[usenames,dvipsnames]{color}
\usepackage{enumitem}

\newcommand{\remove}[1]{}

\newcommand{\name}[1]{\textsc{#1}}

\newcommand{\Oh}{\mathcal{O}}
\newcommand{\OhStar}{\mathcal{O}^{\star}}

\newcommand{\yes}{\textbf{yes}\xspace}
\newcommand{\no}{\textbf{no}\xspace}

\newcommand{\dHS}{\name{$d$-Hitting Set}\xspace}
\newcommand{\VC}{\name{Vertex Cover}\xspace}

\newcommand{\PLC}{\name{Point Line Cover}\xspace}
\newcommand{\LPC}{\name{Line Point Cover}\xspace}

\newcommand{\NPH}{$\mathsf{NP}$-hard\xspace}
\newcommand{\FPT}{\ensuremath{\mathsf{FPT}}\xspace}
\newcommand{\APXH}{$\mathsf{APX}$-hard\xspace}
\newcommand{\conjecture}{\ensuremath{\mathsf{coNP} \nsubseteq \mathsf{NP/poly}}\xspace}
\newcommand{\caveat}{\ensuremath{\mathsf{coNP} \subseteq \mathsf{NP/poly}}\xspace}

\newcommand{\A}{\mathcal{A}}
\renewcommand{\P}{\mathcal{P}}
\renewcommand{\L}{\mathcal{L}}
\newcommand{\Q}{\mathbb{Q}}
\newcommand{\Que}{\mathcal{Q}}

\newenvironment{parnamedefn}[4]{
\par\medskip\noindent\fbox{%
\begin{minipage}[t]{0.985\textwidth}%
\begin{tabular}{p{0.14\textwidth}p{0.80\textwidth}}
    \multicolumn{2}{l}{\name{#1}} \\
        \textsl{Input:} & {#2} \\ 
        \textsl{Question:} & {#3} \\
        \textsl{Parameter:} & {#4} \\%
   \end{tabular}
\end{minipage}}\par\addvspace{0.4\baselineskip}
}

\theoremstyle{plain} 
\newtheorem{theorem}{Theorem}
\newtheorem{corollary}{Corollary}

\newtheorem{lemma}{Lemma}

\newtheorem{claim}{Claim} 

\theoremstyle{remark} 
\newtheorem{remark}{Remark} 

\theoremstyle{definition} 
\newtheorem{definition}{Definition}

\newif\iffullpaper
\fullpapertrue

\begin{document}

\title{Point Line Cover: The Easy Kernel is Essentially Tight}
\titlerunning{Point Line Cover: The Easy Kernel is Essentially Tight}
\author{Stefan Kratsch\inst{1}\thanks{Supported by the German Research Foundation (DFG), research project PREMOD, KR~4286/1.}
\and Geevarghese Philip\inst{2}\thanks{Supported by the
  Indo-German Max Planck Center for Computer Science(IMPECS).}
\and Saurabh Ray\inst{3}\thanks{Work done in part while at Max-Planck-Institut f\"{u}r Informatik, Saarbr\"{u}cken, Germany.}
}

\institute{
Technical University Berlin, Germany, \email{stefan.kratsch@tu-berlin.de}
\and
Max-Planck-Institut f\"{u}r Informatik,
Saarbr\"{u}cken, Germany, \email{gphilip@mpi-inf.mpg.de}
\and
Ben-Gurion University, Israel, \email{saurabh@bgu.ac.il}
}

\authorrunning{Kratsch, Philip, and Ray}

\clearpage

\maketitle
\begin{abstract}
  The input to the \NPH \PLC problem (PLC) consists of a set \(\P\) of
  \(n\)~points on the plane and a positive integer \(k\), and the
  question is whether there exists a set of at most \(k\) lines
  which pass through all points in \(\P\).  By straightforward
  reduction rules one can efficiently reduce any input to one with
  at most \(k^{2}\) points. We show that this easy reduction is
  already essentially tight under standard assumptions. More
  precisely, unless the polynomial hierarchy collapses to its
  third level, for any \(\varepsilon>0\), there is no polynomial-time algorithm that reduces
  every instance \((\P,k)\) of PLC to an equivalent instance with
  \(\Oh(k^{2-\varepsilon})\) points.
  This answers, in the negative, an open problem posed by
  Lokshtanov (PhD Thesis, 2009).

  Our proof uses the notion of a kernel from parameterized complexity, and the
  machinery for deriving lower bounds on the \emph{size} of kernels
  developed by Dell and
  van~Melkebeek (STOC 2010). It has two main ingredients: We first
  show, by reduction from \VC, that---unless the polynomial hierarchy
  collapses---PLC has no kernel of \emph{total size}
  \(\Oh(k^{2-\varepsilon})\) bits. This does not directly imply the claimed
  lower bound on the \emph{number of points}, since
  the best known polynomial-time encoding of a PLC instance with
  \(n\) points requires \(\omega(n^{2})\) bits. To get around this
  hurdle we build on work of Goodman, Pollack and Sturmfels (STOC
  1989) and devise an \emph{oracle communication protocol} of cost
  \(\Oh(n\log{}n)\) for PLC; its main building block is a bound of
  \(\Oh(n^{\Oh(n)})\) for the order types of $n$ points that are not 
  necessarily in general position and an explicit (albeit slow)
  algorithm that enumerates all possible order types of $n$ points. 
  This protocol, together with the
  lower bound on the total size (which also holds for such
  protocols), yields the stated lower bound on the number of
  points.
  
  While a number of essentially tight polynomial lower bounds on total sizes
  of kernels are known, our result is---to the best of our
  knowledge---the first to show a nontrivial lower bound for
  structural/secondary parameters.
\end{abstract}

\newpage
\section{Introduction}\label{sec:intro}
Recall that a point \((a,b)\) in the two-dimensional plane is said
to \emph{lie on} a line \(y=mx+c\) if and only if \(b=ma+c\)
holds.\footnote{Throughout this work we assume the implicit
  presence of an arbitrary but fixed Cartesian coordinate
  system.} In this case we also say that the line \emph{covers}---or
\emph{passes through}---the point, and also, symmetrically, that
the point \emph{covers} the line. The \PLC (PLC) problem of finding a
smallest number \(\ell\) of lines which cover a given set of \(n\)
points on the plane is motivated by various practical
applications~\cite{HassinMegiddo1991,MegiddoTamir1982}.  Megiddo
and Tamir~\cite{MegiddoTamir1982} showed that the problem is
\NPH. Kumar et al.~\cite{KumarAryaRamesh2000} and Brod{\'e}n et
al.~\cite{BrodenHammarNilsson2001} showed that the problem is
\APXH, and therefore cannot be approximated to within an
arbitrarily small constant factor in polynomial time, unless
\textsf{P=NP}. Grantson and
Levcopoulos~\cite{GrantsonLevcopoulos2006} devised an
\(\Oh(\log\ell)\)-factor approximation algorithm which runs in
polynomial time when \(\ell\in\Oh(n^{1/4})\).

In this work we study the complexity of the decision version
where, given a set~$\P$ of~$n$ points and an integer~$k$, the task
is to determine whether the points in~$\P$ can be covered by
\emph{at most~$k$ lines}. While there is no a priori relation
between~$n$ and~$k$, it is known that one can efficiently reduce
to the case where~$\P$ contains at most~$k^2$ points: First, any
line containing at least~$k+1$ points is \emph{mandatory} if we
want to use at most~$k$ lines in total. Indeed, the $k+1$ points on
such a line would otherwise require~$k+1$ separate lines since any
two lines share at most one point. Thus we may delete any such
line and all points which the line covers, and decrease~$k$ by one
without changing the outcome (\yes or \no). Second, if no line
contains more than~$k$ points, then~$k$ lines can cover at
most~$k^2$ points. Thus, if~$n>k^2$ then we return \no, and else
we have~$n\leq k^2$ as claimed. It is a natural and, in our
opinion, very interesting question whether this simple reduction
process can be improved to yield a significantly better reduction,
e.g., to~$k^{2-\varepsilon}$ points for some~$\varepsilon>0$. This
has been posed as an open problem by
Lokshtanov~\cite{LokshtanovThesis2009}.

Our main result is a negative answer to Lokshtanov's question.

\begin{theorem}\label{theorem:plc:pointlowerbound}
  Let~$\varepsilon>0$.  Unless \caveat, there is no
  polynomial-time algorithm that reduces every
  instance~$(\P,k)$ of \PLC to an equivalent instance
  with~$\Oh(k^{2-\varepsilon})$ points.\footnote{The assumption
    that \conjecture is backed up by the fact that \caveat is
    known to imply a collapse of the polynomial hierarchy to its
    third level~\cite{Yap1983}; it is therefore widely believed.}
\end{theorem}

The first part of our result is a lower bound
of~$\Oh(k^{2-\varepsilon})$ on the \emph{size} of a \PLC instance
which can be obtained by efficient (reduction) algorithms. 
For this we use
established machinery from parameterized complexity (formal
definitions and tools are given in
Section~\ref{sec:preliminaries}).

\begin{theorem}\label{theorem:plc:sizebound}
  Let~$\varepsilon>0$. The \PLC problem has no kernel of size
  \(\Oh(k^{2-\varepsilon})\), i.e., no polynomial-time reduction
  to equivalent instances of size~$\Oh(k^{2-\varepsilon})$, unless
  \caveat.
\end{theorem}

To get the much more interesting lower bound of
\autoref{theorem:plc:pointlowerbound} on the number of points we
have to relate the number of points in an instance of PLC to the
size of the instance, as tightly as possible. The catch here is
that there is \emph{no known efficient and sufficiently tight}
encoding of PLC-instances with~$n$ points into a small---in terms
of~$n$---number of bits~\cite[Section 1]{GPS89}. Further, known
tools from parameterized complexity only provide lower bounds on
total \emph{sizes} of kernels. We throw more light on this
distinction in our brief discussion on \VC later in this section.

Getting around this hurdle is the second component of our result:
It is known that point sets \emph{in general position}, i.e., with
no three points sharing a line, can be partitioned
into~$2^{\Oh(n\log{}n)}$ equivalence classes (later, and formally,
called \emph{order types}) in a combinatorial sense; we extend
this fact to arbitrary sets of~$n$ points (\autoref{lem:otub}).
The answer (\yes or \no) to a PLC instance is determined by the
order type of points in the instance. Thus, morally, the relevant
information expressed by the coordinates of the~$n$ points is
only~$\Oh(n\log n)$ bits. However, no efficient numbering scheme
for order types is known, and hence we do not know how to
efficiently compute such a representation with~$\Oh(n\log n)$
bits. Our solution for this is as follows: First, we prove that
order types of point sets are decidable, and devise a
computable---albeit inefficient---enumeration scheme, which we
believe to be of independent interest.

\begin{lemma}\label{lem:enumerate_order_types}
  There exists an algorithm which enumerates, for each
  \(n\in\mathbb{N}\), all order types 
  defined by \(n\) points in the plane.
\end{lemma}

We then use the fact that the size lower bound of
\autoref{theorem:plc:sizebound} for \PLC also holds in the setting
of a two-player communication protocol where the first player
holds the input instance but can only communicate a bounded number
of bits to an all-powerful oracle that will solve the instance
(see \autoref{sec:preliminaries} for definitions). Using our
enumeration scheme, the oracle can query the first player about
her instance and, effectively, perform a binary search for the
order type of the input. This kind of an oracle communication
protocol was introduced by Dell and
van~Melkebeek~\cite{DellvanMelkebeek2010}, who developed lower
bound techniques for such protocols; they also used such an
\emph{active} oracle to answer a question about sparse languages.

\paragraph{Related work.} Langerman and
Morin~\cite{LangermanMorin2005} showed that \PLC is
\emph{fixed-parameter tractable} (\FPT) with respect to the target
number~$k$ of lines by giving an algorithm that runs in time
\(\OhStar(k^{2k})\); we use~$\OhStar$-notation to hide factors
polynomial in the input size. Later, Grantson and
Levcopoulos~\cite{GrantsonLevcopoulos2006} proposed a faster
algorithm which solves the problem in time
\(\OhStar((k/2.2)^{2k})\). The fastest \FPT algorithm currently
known for the problem is due to Wang et al.~\cite{WangLiChen2010},
and it solves the problem in time \(\OhStar((k/1.35)^{k})\).

Langerman and Morin~\cite{LangermanMorin2005} also gave the
reduction to at most $k^2$ points that we outlined earlier.  In
terms of parameterized complexity this can be seen to lead to a
\emph{polynomial kernelization}, using a few standard arguments:
Clearly, we are effectively asking for a set cover for the point
set (where sets are given by the lines). There are at most $k^4$
lines and we can encode the points contained in each line using at
most $k^2$ bits; this uses \(\Oh(k^{6})\) bits. Since \textsc{Set
  Cover} is in \textsf{NP}, there is a Karp reduction back to
\PLC, which gives an equivalent PLC instance of size polynomial
in~$k^6$, i.e., polynomial in~$k$.  Estivill-Castro et
al.~\cite{Estivill-CastroHeednacramSuraweera2009} describe a
couple of additional reduction rules which, while not improving
beyond~$\Oh(k^2)$ points, yield kernels which seem to be better
amenable to faster subsequent processing by algorithms. See also
the PhD thesis of Heednacram~\cite[Chapter 2]{Heednacram2010} for
a survey of related work.

A number of papers build on the work of Dell and van
Melkebeek~\cite{DellvanMelkebeek2010} to prove concrete polynomial
lower bounds for the \emph{size} of kernelizations for certain
problems, e.g.,~\cite{DellM12,HermelinW12,CyganGH13_arxiv}; to the
best of our knowledge none of them obtain tight bounds of
secondary parameters (other than the number of edges, which is
usually an immediate consequence).  For other recent developments
in kernelization we refer to the survey of Lokshtanov et
al.~\cite{LokshtanovMS12}.

\paragraph{Vertex cover.} 
To shed more light on the distinction between lower bounds on
kernel \emph{sizes} and lower bounds on other, secondary or
structural, parameters, let us briefly recall what is known for
the well-studied \VC problem: It is known how to efficiently
reduce \VC input instances to equivalent instances with $\Oh(k)$
vertices and $\Oh(k^2)$ edges~\cite{ChenKJ01}, and hence to size
$\Oh(k^2)$ by adjacency matrix encoding. The breakthrough work of
Dell and van Melkebeek \cite{DellvanMelkebeek2010} showed that no
reduction to size $\Oh(k^{2-\varepsilon})$ is possible for any
$\varepsilon>0$, unless \caveat. Since graphs with $m$ edges (and
no isolated vertices) can be represented with $\Oh(m\log m)$ bits
this also implies that no reduction to $\Oh(k^{2-\varepsilon})$
\emph{edges} is possible. Similarly, we cannot get to
$\Oh(k^{1-\varepsilon})$ vertices since that would allow an
adjacency matrix encoding of size
$\Oh(k^{2-2\varepsilon})$. Contrast this with \dHS, i.e., \VC on
hypergraphs with edges of size $d$, for which the best known
reduction achieves $\Oh(k^{d-1})$ vertices and $\Oh(k^d)$
hyperedges \cite{Abu-Khzam10}; Dell and van Melkebeek
\cite{DellvanMelkebeek2010} ruled out size
$\Oh(k^{d-\varepsilon})$. Thus the bound on the number of edges is
essentially tight, but the implied lower bound on the number of
vertices is much weaker, namely $\Oh(k^{1-\varepsilon})$. The
takeaway message is that lower bounds for secondary parameters
can, so far, be only concluded from their effect on the instance
size; there are no ``direct'' lower bound proofs.

\paragraph{Organization of the rest of the paper}
In the next section we state various definitions and preliminary
results.  In \autoref{sec:order_types} we prove an upper bound of
\(n^{\Oh(n)}\) on the number of combinatorially distinct order
types of \(n\) points in the plane, and show that there is an
algorithm which enumerates all order types of \(n\) points,
thereby proving \autoref{lem:enumerate_order_types}. In
\autoref{sec:reduction} we present our reduction from \VC to \PLC
which proves a more general claim than
\autoref{theorem:plc:sizebound}. We describe an oracle
communication protocol of cost \(\Oh(n\log{}n)\) for \PLC in
\autoref{section:pointlowerbound}, and this yields our main
result, \autoref{theorem:plc:pointlowerbound}). We conclude in
\autoref{sec:conclusion}.

\section{Preliminaries}\label{sec:preliminaries}
We use $[n]$ to denote the set \(\{1,2,\ldots,n\}\). Throughout
the paper we assume the presence of an arbitrary but fixed
Cartesian coordinate system. All geometric objects are referenced
in the context of this coordinate system. We use $p_x$ and $p_y$
to denote the $x$ and $y$ coordinates, respectively, of a point
$p$. A set of points in the plane is said to be \emph{in general
  position} if no three of them are collinear; a set of points
which is not in general position is said to be
\emph{degenerate}. For two points \(p\neq{}q\) in the plane, we
use \(\overline{pq}\) to denote the unique line in the plane which
passes through \(p\) and \(q\); we say that the line
\(\overline{pq}\) is \emph{defined by} the pair \(p,q\).

Let $a,b,c$ be three points in the plane. We say that the
\emph{orientation} of the ordered triple $\langle a,b,c \rangle$
is $+1$ if the points lie in counter-clockwise position, $-1$ if
they lie in clockwise position and $0$ if they are collinear.
Formally, let
\[M(\langle a,b,c \rangle) = \left(%
\begin{array}{ccc}%
1 & a_x & a_y\\%
1 & b_x &b_y\\%
1 & c_x & c_y%
\end{array}%
\right).%
\]
Then,
\(orientation(\langle{}a,b,c\rangle)=sgn\,det\,M(\langle{}a,b,c\rangle)\)
where $sgn$ is the sign function and $det$ is the determinant
function. Note that the determinant above is zero if and only if
the rows are linearly dependent in which case, without loss of
generality, $\langle 1, a_x, a_y \rangle = \lambda \langle 1, b_x,
b_y\rangle + \mu \langle 1, c_x, c_y\rangle$. Comparing the first
coordinates on both sides of the inequality we see that $\mu =
1-\lambda$ which is equivalent to saying that one of the points is
a convex combination of the other two. Hence
\(orientation(\langle{}a,b,c\rangle) \) is zero exactly when $a$,
$b$, and $c$ are collinear.

Let \(\P=\langle{}p_{1},\cdots,p_{n}\rangle\) be an ordered set of
points, where $p_i=(x_i, y_i)=\P[i]$.  Denote by $\binom{[n]}{3}$
the set of ordered triples $\langle i,j,k \rangle$ where $i<j<k$
and $i,j,k \in [n]$. Define $\sigma: \binom{[n]}{3} \mapsto \{
+1,0,-1\}$ to be the function $\sigma (\langle i,j,k \rangle) =
orientation(p_i,p_j,p_k)$.  The function $\sigma$ is called the
\emph{order type} of $\P$. Observe that the order type of a point
set depends on the order of points and not just on the set of
points. Two point sets \(\P,\Que\) of the same size \(n\) are said
to be \emph{combinatorially equivalent} if there exist orderings
\(\P'\) of \(\P\) and \(\Que'\) of \(\Que\) such that the order
types of \(\P'\) and \(\Que'\)---which are both functions of type
\(\binom{[n]}{3}\mapsto\{+1,0,-1\}\)---are identical. Otherwise we
say that \(\P\) and \(\Que\) are \emph{combinatorially
  distinct}. If two order types come from combinatorially
equivalent (distinct) point sets, we call the order types
combinatorially equivalent (distinct).

It is not difficult to see that combinatorial distinction is a
correct criterion for telling non-equivalent instances of \PLC
apart.

\begin{lemma}\label{lemma:comb_equiv_plc}
  Let \((\P,k),(\Que,k)\) be two instances of \PLC. If the point
  sets \(\P,\Que\) are combinatorially equivalent, then \((\P,k)\)
  and \((\Que,k)\) are equivalent instances of \PLC.
\end{lemma}

\begin{proof}
  Let \(\P,\Que\) be combinatorially equivalent, let
  \(|\P|=|\Que|=n\), and let \(\P',\Que'\) be orderings of \(\P\)
  and \(\Que\), respectively, with identical order types. Observe
  first that the combinatorial equivalence provides us with a
  natural bijection \(\pi:\P\mapsto\Que\), defined as follows. Let
  \(p\in\P\), and let \(i\in[n]\) be such that \(\P'[i]=p\). Then
  \(\pi(p)=\Que'[i]\).

  For any subset $T \subseteq \P $, let $\pi(T)$ denote the set
  $\{\pi(t) : t\in T\}$. Let $S\subseteq \P$ be collinear.  For
  any triple $a,b,c \in \pi(S)$, $orientation(\langle a, b,
  c\rangle) = $ $orientation(\langle \pi^{-1}(a), \pi^{-1}(b),
  \pi^{-1}(c)\rangle)$ $ = 0$ where the first equality follows
  from the combinatorial equivalence of $\P$ and $\Que$ and the
  second equality follows from the collinearity of every triple of
  points in $S$. This implies that every triple of points in
  $\pi(S)$ are collinear, which is equivalent to saying that
  $\pi(S)$ is a collinear subset of $\Que$. Similarly, since $\pi$
  is a bijection, if $\pi(S)$ is collinear for some $S \subseteq
  \P $ then $S$ is also collinear.  Thus, $S$ is a collinear
  subset of $\P$ if and only if $\pi(S)$ is a collinear subset of
  $\Que$.

  Let \((\P,k)\) be a \yes instance, and let \(\L\) be a set of at
  most \(k\) lines which cover all points in \(\P\). Without loss
  of generality, each of the lines in $\L$ passes through at least
  two points in $\P$ since we can always replace a line through a
  single point by a line through two or more points.  For each
  $\ell \in \L$, denote by $S_\ell$ the subset of points of $\P$
  that $\ell$ covers. Since $S_\ell$ is collinear, so is
  $\pi(S_\ell)$ and thus we can define $\ell'$ to be the line
  through $\pi(S_\ell)$. Then, $\L' = \{\ell' : \ell \in \L\}$
  covers $\Que$ since for every $q \in \Que$ \, there is line
  $\ell\in\L$ that covers $\pi^{-1}(q)$. This implies that $(\Que,
  k)$ is a \yes instance. Again, since $\pi$ is a bijection, we
  have that if $(\Que, k)$ is a \yes instance then $(\P,k)$ is a
  \yes instance. Thus, $(\P, k)$ is a \yes instance if and only if
  $(\Que,k)$ is a \yes instance.
\end{proof}

\begin{remark}
  For the previous lemma, it is not important that a triple in
  $\P$ has the same orientation as the corresponding triple (under
  the bijection $\pi$) in $\Que$. It is sufficient to ensure that
  a triple in $\P$ has orientation $0$ iff the corresponding
  triple has orientation $0$. We could define and use a coarser
  notion of equivalence based on this but we choose to stick to
  order types since they are well studied.
\end{remark}
We need the following straightforward polynomial-time construction
of point sets with some special properties for our reduction in
\autoref{sec:reduction}:

\begin{lemma}\label{lemma:special_point_set_construction}
  For any positive integer \(n\) we can construct, in time
  polynomial in \(n\), a set \(\P=\{p_{1},\ldots,p_{n}\}\) of
  \(n\) points with the following properties:
  \begin{enumerate}
  \item The set of points \(\P\) is in general position;
    \item No two lines defined by pairs of points in $\P$ are
      parallel; and,
    \item No three lines defined by pairs of points in $\P$ pass
      through a common point \emph{outside} $\P$.
  \end{enumerate}
\end{lemma}
\begin{proof}
  We describe a polynomial-time construction of such a point
  set. The construction is trivial when $n=1$. Assume that a set
  \(\P'\) of \(1\leq{}t<n\) points with these properties has been
  constructed. Observe that all the points which are ``forbidden''
  by the set \(\P'\)---that is, points which cannot be added to
  the set \(\P'\) without violating one of the three
  properties---lie on a bounded number of lines. For instance, the
  first property is violated if and only if a new point lies on
  one of the $\binom{t}{2}$ lines defined by pairs of points in
  \(\P'\).  Similarly, the second and third conditions forbid
  points which lie on sets of $\Oh(t^3)$ and $\Oh(t^5)$ lines,
  respectively. Further, we can compute all these lines in
  polynomial time by enumerating all possibilities.

  We augment the set \(\P'\) by choosing a point which is not on
  one of these lines.  To facilitate this we pick all our points
  from a grid of size $n^6 \times n^6$.  As we argued above, only
  $\Oh(n^5)$ lines are forbidden at any point during the
  construction. Each of these lines intersects the
  \(n^{6}\times{}n^{6}\) grid in $\Oh(n^6)$ points. Thus at most
  $\Oh(n^{11})$ of the $n^{12}$ grid points are forbidden at any
  time, and we augment \(\P'\) with a point on the grid which is
  not forbidden. Our construction consists of repeating this step
  \(n\) times, and takes polynomial time.
\end{proof} 

\paragraph{Graphs.}
All graphs in this article are finite, simple, undirected, and
loopless. In general we follow the graph terminology of Diestel's
textbook on the subject~\cite{Diestel2005}. The (open)
neighborhood \(N_{G}(v)\) of a vertex \(v\) in a graph \(G\) is
the set of all vertices \(u\) of \(G\) such that \(u\) and \(v\)
are adjacent. A \emph{vertex cover} of a graph \(G\) is a subset
\(S\subseteq V(G)\) of the vertex set of \(G\) such that for every
edge \(e\) of \(G\), there is at least one vertex in \(S\) which
is incident with \(e\).

\paragraph{Parameterized complexity.}
A parameterized problem \(\mathcal{A}\) is a subset of
\(\Sigma^*\times\mathbb{N}\) for some finite alphabet \(\Sigma\);
the second component of
instances~\((x,k)\in\Sigma^*\times\mathbb{N}\) is called the
\emph{parameter}.  A \emph{kernelization algorithm} for problem
\(\A\) is an algorithm that given \((x,k)\in\Sigma^*\) takes time
polynomial in \(|x|+k\) and outputs an equivalent instance
\((x',k')\) of \(\A\)---that is, one which preserves the \yes/\no
answer---such that \(|x'|,k'\) are both bounded by some
function~$h$ of \(k\); this function is called the \emph{size} of
the
kernelization~\cite{DowneyFellows1999,FlumGrohe2006,Niedermeier2006}. If~$h$
is polynomially bounded then we have a \emph{polynomial
  kernelization}. One of the most important (fairly) recent
results in kernelization is a framework for ruling out polynomial
kernels for certain problems, assuming that \conjecture
\cite{BodlaenderDFH09,FortnowS11}.

Our work uses an extension of the lower bound framework that is
due to Dell and van Melkebeek~\cite{DellvanMelkebeek2010}. Their
work provided the first tool to prove concrete polynomial lower
bounds on the possible size of kernelizations. In fact, their
lower bounds are proven for an abstraction of kernelization as
oracle communication protocols of bounded cost; the lower bounds
carry over immediately.  (For intuition, the first player could
run a kernelization with size $h$ and send the outcome to the
oracle, who decides the instance, obtaining a protocol of
cost~$h(k)$ for deciding~$(x,k)\in\A$. Of course, having a
multi-round communication or using an active oracle that queries
the user appears to be more general than kernelization.)

\begin{definition}[Oracle Communication
  Protocol]\label{definition:oracleprotocol}
  An oracle communication protocol for a language \(L\) is a
  communication protocol between two players. The first player is
  given the input \(x\) and has to run in time polynomial in the
  length of the input; the second player is computationally
  unbounded but is not given any part of \(x\). At the end of the
  protocol the first player should be able to decide whether
  \(x\in{}L\). The \emph{cost} of the protocol is the number of
  bits of communication from the first player to the second
  player.
\end{definition}

Our lower bound for the \emph{size} of kernels for \PLC will follow from
a reduction from the well-known \NPH \VC problem (cf.~\cite{Karp1972}).
We recall the problem setting.

\begin{parnamedefn} {\VC}
  {A graph \(G\), and a positive integer
    \(k\).}
  {Does the graph \(G\) have a vertex cover of size at most \(k\)?}
  {\(k\)}
\end{parnamedefn}

The smallest known kernel for the problem has at most \((2k-c\log
k)\) vertices for any fixed constant \(c\)~\cite{Lampis2011}. This
kernel can have \(\Omega(k^{2})\) \emph{edges}, and so the total
\emph{size} of the kernel is \(\Omega(k^{2})\), and not
\(\Oh(k)\). Dell and van~Melkebeek~\cite{DellvanMelkebeek2010}
showed that this is---in a certain sense---the best possible upper
bound on the kernel size for \VC. In fact, they proved much more
general lower bounds about the cost of a communication process
used to decide languages.

\begin{theorem}[\textbf{\cite[Theorem 2]{DellvanMelkebeek2010}}]\label{theorem:vc_protocol_lower_bound}
  Let \(d\geq 2\) be an integer, and \(\varepsilon\) a positive
  real. If \conjecture, then there is no protocol of cost
  \(\Oh(n^{d-\varepsilon})\) to decide whether a \(d\)-uniform
  hypergraph on \(n\) vertices has a vertex cover of at most \(k\)
  vertices, even when the first player is conondeterministic.
\end{theorem}

We will use the immediate corollary that
\VC admits no oracle communication protocol of cost \(\Oh(k^{2-\varepsilon})\)
for deciding whether a graph $G$ has a vertex cover of size at most $k$
(and hence also no kernelization of that size), for any $\varepsilon>0$, unless \caveat.

\paragraph{The \PLC problem.}
A formal statement of the central problem of this paper is as
follows:
\begin{parnamedefn} {\PLC}
  {A set of \(n\) points
    \(\P\) in the plane, and a positive integer
    \(k\).}
  {Is there a set of at most \(k\) lines in the
    plane which cover all the points in \(\P\)?}
  {\(k\)}
\end{parnamedefn}

We use the following ``dual'' problem to \PLC in our proof of
\autoref{theorem:plc:sizebound}. In the input to this problem,
each line in the set \(\L\) is given as a pair
\((m,c)\subseteq\Q\times\Q\) where \(m\) is the slope of the line
and \(c\) its \(Y\)-intercept.

\begin{parnamedefn} {\LPC}
  {A set of \(n\) lines
    \(\L\) in the plane, and a positive integer
    \(k\).}
  {Is there a set of at most \(k\) points in the
    plane which cover all the lines in \(\L\)?}
  {\(k\)}
\end{parnamedefn}
There is a polynomial-time, parameter-preserving reduction from
the dual to the primal: 

\begin{lemma}\label{lem:lpc_to_plc}
  There is a polynomial-time reduction from the \LPC problem to
  the \PLC problem which preserves the parameter \(k\). That is,
  the reduction takes an instance \((\L,k)\) of \LPC to an
  equivalent instance \((\P,k)\) of \PLC.
\end{lemma}

\begin{proof}
  The lemma follows from a well known duality of point and lines
  (see e.g.,~\cite{M4}). Given a set of points $\P =
  \{p_1,\cdots,p_n\}$ and a set of lines $\L = \{ \ell_1, \cdots,
  \ell_m \}$, one can obtain lines $\tilde{\P} =\{\tilde{p}_1,
  \cdots, \tilde{p}_n\}$ and points $\tilde{\L}=\{\tilde{\ell}_1,
  \cdots, \tilde{\ell}_m\}$ such that the point $\tilde{\ell}_i$
  lies on the line $\tilde{p}_j$ if and only if the point $p_j$
  lies on the line $\ell_i$. In other words, each line can be
  replaced by a point and each point can be replaced by a line
  while preserving the incidences between lines and points.

  To see this, first note that it can be assumed without loss of
  generality that none of the given lines $\ell_i$ is
  vertical. This can be easily done by a suitable rotation of the
  coordinate system since there are only a finite number of
  directions for the $y$-axis to avoid. Next note that a point
  $(a,b)$ lies on a line $y= mx + c$ if and only if the point
  $(m,-c)$ lies on the line $y=ax-b$ since both conditions are
  equivalent to $b=ma+c$. Thus if $p_i = (a_i, b_i)$, we can take
  $\tilde{p}_i$ to be the line $y = a_ix -b_i$ and if $\ell_j$ is
  the line $y = m_jx + c_j$, we can take $\tilde{\ell}_j$ to be
  the point $(m_j,-c_j)$. This transformation can be done in
  polynomial time for any reasonable representation of rational
  numbers.

  The above ``point-line duality'' implies that $(\L,k)$ and
  $(\tilde{\L},k)$ are equivalent instances of \LPC and \PLC respectively.
  \end{proof}

\section{Enumerating Order Types}\label{sec:order_types}
Goodman and Pollack~\cite{GoodmanPollack1991} proved that the
number of combinatorially distinct order types on $n$ points in
general position is $\Oh(n^{4n(1+o(1))})$.  Using this we prove an
upper bound on the number of all order types on $n$ points,
including those that come from point sets not in general
position. We defer the proof of the following lemma to the end of
this section.

\begin{lemma}\label{lem:otub}
  There are at most $n^{\Oh(n)}$ combinatorially distinct order
  types defined by $n$ points in $\mathbb{R}^2$.
\end{lemma}

Thus there are only $2^{\Oh(n \log{n})}$ \emph{combinatorially
  distinct} instances of the \PLC problem with $n$ points---see
\autoref{lemma:comb_equiv_plc}.  In this section we consider the
algorithmic problem of enumerating all order types of $n$-point
sets; we need such an algorithm for proving our lower bound in
\autoref{section:pointlowerbound}. Note that we do not need an
\emph{efficient} algorithm; we only need to show that there is an
algorithm which solves the problem and terminates in finite time
for each \(n\).  Given this, a first approach would be to consider
a large enough $N=f(n)$ and produce the order type of every set of
$n$ points whose coordinates are integers in $[N]$. The hope here
would be that every order type comes from some point set whose
points have coordinates of finite precision. Unfortunately, this
is not true: there are order types $\sigma$ which cannot be
realized by point sets with rational coordinates~\cite{GPS89}. We
need a somewhat more sophisticated argument to prove
\autoref{lem:enumerate_order_types}.

\begin{proof}[Proof of \autoref{lem:enumerate_order_types}]
  Goodman et
  al.~\cite{GPS89} show that for any function \(\sigma\) which is
  the order type of some (unknown) set of \(n\) points \emph{in
    general position}, there exists an ordered set \(\P\) of \(n\)
  points in the plane such that (i) the order type of \(\P\) is
  \(\sigma\), and (ii) each point in \(\P\) has integer
  coordinates in \(\left[2^{2^{\Oh(n)}}\right]\).  We show that by
  modifying their proof to work with the more general notion of an
  order type of \emph{cells} rather than points we can
  \emph{enumerate} all order types of \(n\) points, including
  those which correspond to degenerate point sets.

  We start with a brief overview of the proof of Goodman et al.
  We are given a function
  \(\sigma:\binom{[n]}{3}\mapsto\{-1,+1\}\) which is the order
  type of some point set in general position. The goal is to show
  that there exists an ordered point set
  \(\P=\langle{}p_{1},\ldots,p_{n}\rangle\) with order type
  $\sigma$ such that the coordinates of the points \(p_{i}\) are
  integers of magnitude $2^{2^{\Oh(n)}}$.  We first find an
  ordered point set \(\Que=\langle{}q_{1},\ldots,q_{n}\rangle\)
  whose order type is $\sigma$---but which may have points with
  non-integral coordinates---and then modify \(\Que\) to obtain
  \(\P\).  Let $x_i$ and $y_i$ be variables which stand for the
  $x$ and $y$ coordinates of $q_i$.  Let
  \(\Delta(\langle{}i,j,k\rangle)=det\ M(\langle{}
  q_{i},q_{j},q_{k}\rangle)\).  Our variables must satisfy the set
  of equations
  \(sgn\,\Delta(\langle{}i,j,k\rangle)=\sigma(i,j,k)\;\forall\,\langle{}i,j,k\rangle\in\binom{[n]}{3}\),
  which can be rewritten as
  \(\Delta(\langle{}i,j,k\rangle)\cdot\sigma(i,j,k)>0\;\forall\,\langle{}i,j,k\rangle\in\binom{[n]}{3}\). This
  set of equations can safely\footnote{See the paper of Goodman et
    al.~\cite{GPS89} for all technical details.} be replaced by
  \(\Delta(\langle{}i,j,k\rangle)\cdot\sigma(i,j,k)\geq1\;\forall\,\langle{}i,j,k\rangle\in\binom{[n]}{3}\). Since
  we started with the assumption that $\sigma$ is the order type
  of some point set, we know that this set of inequalities has a
  solution. Observe that this is a set of inequalities of degree
  two where the coefficients have bit length $1$. Hence, by a
  result of Grigor'ev and Vorobjov~\cite{GrigorevVorobjov1988}
  they have a solution \(\Que=\langle q_{1},\ldots,q_{n}\rangle
  \in\mathbb{R}^{2n}\) within a ball of radius $2^{2^{\Oh(n)}}$
  around the origin.  It can be shown that the minimum distance
  between a point in \(\Que\) and the line determined by two other
  points in \(\Que\) is bounded away from $0$, and that the
  minimum angle defined by three points in \(\Que\) is bounded
  away from $0$. This is because (i) the determinants in our
  inequalities have magnitude at least $1$, and (ii) the points in
  \(\Que\) lie inside a ball of bounded radius. It follows that
  moving these points by small distances does not change their
  order type, and so each of these points can be moved to a point
  whose coordinates are integral multiples of some
  \(\ell\in2^{-2^{\Oh(n)}}\). By scaling up we get integral
  coordinates of magnitude at most $2^{2^{\Oh(n)}}$. This gives us
  the required set $\P$.

  This strategy breaks down when $\sigma$ is an order type of a
  \emph{degenerate} point set. To see why, let us consider the
  order type \(\sigma\) of a degenerate point set and proceed
  exactly as above. Since we now have collinear triples, our set
  of inequalities has the following form
\begin{equation}\label{eqn:degenerate_inequalities}
\begin{aligned}
\Delta(\langle{}i,j,k\rangle)\cdot\sigma(i,j,k)&\geq1\;\;\forall\,\langle{}i,j,k\rangle\in\binom{[n]}{3}\text{ s.t. }\sigma(i,j,k)\neq0\\
\Delta(\langle{}i,j,k\rangle)&=0\;\;\forall\,\langle{}i,j,k\rangle\in\binom{[n]}{3}\text{ s.t. }\sigma(i,j,k)=0,
\end{aligned}
\end{equation}
and the Grigor'ev-Vorobjov bound still holds.  The problem arises
in the step where we move some point slightly to make its
coordinates multiples of \(\ell\in2^{-2^{\Oh(n)}}\). While this
cannot make a clockwise triple counter-clockwise or \emph{vice
  versa}, it could easily destroy the collinearity of triples and
thus violate the set of equations. 

To get around this problem, we replace points with bigger regions
in the argument. We first extend the notion of orientation from
triples of points to triples of convex regions in the plane, as
follows. Let $r_a$, $r_b$, and $r_c$ be three convex regions in
the plane. We assign an orientation $+1$ (resp. $-1$) to the
triple $\langle r_a,r_b,r_c \rangle$ if for each choice of points
$p_a \in r_a, p_b \in r_b$, and $p_c \in r_c$ the triple $\langle
p_a, p_b, p_c \rangle $ has orientation $+1$ (resp.\ $-1$). If a
triple $\langle r_a,r_b,r_c \rangle$ is not assigned an
orientation according to this rule then we assign it the
orientation $0$; this happens if and only if there is a line
intersecting the three regions. As with points, we define an order
type for an ordered set of regions $\langle r_1, \cdots,
r_n\rangle$ to be the function \(\sigma\) which satisfies
$\sigma(\langle i,j, k \rangle) = orientation(r_i, r_j, r_k)$.

Once we have a solution
\(\Que=\langle{}q_{1},\ldots,q_{n}\rangle\in\mathbb{R}^{2n}\) to
the inequalities (\ref{eqn:degenerate_inequalities}) as described
above, we superimpose a grid with cells of side length
$2^{-2^{\Oh(n)}}$ onto the coordinate system. For each
\(q_{i}\in\Que\), we take $r_i$ to be the cell in which the point
\(q_{i}\) lies. Observe that the order type of these $r_i$'s is
the same as the order type of the $q_i$'s. This can be seen as
follows. Whenever a triple $ \langle q_{i_1}, q_{i_2}, q_{i_3}
\rangle$ of points in $\Que$ have non-zero orientation, any three
points $a, b, c$ lying in the cells $r_{i_1}, r_{i_2}$, and
$r_{i_3}$ respectively have the same orientation since, as argued
before, moving the points slightly does not change
orientation. Also, if a triple of points has orientation $0$ then
the corresponding cells have orientation $0$ by definition. As
before, we can scale up the grid so that the $r_i$'s are cells of
side length $1$ whose corners have integer coordinates.  It
follows that we can enumerate all order types defined by $n$
points by (i) taking an integer grid of size $2^{2^{\Oh(n)}}$, and
(ii) listing all order types defined by some $n$-subset of the
cells in this grid.
\end{proof}

\begin{proof} [Proof of \autoref{lem:otub}]
  For any given order type defined by an ordered point set
  \(P=\langle{}p_{1},\dotsc,p_{n}\rangle\), we show that there is
  an ordered set of $2n$ points in general position whose order
  type \emph{encodes} the order type of $P$. Together with the
  Goodman-Pollack bound, this shows that the number of
  combinatorially distinct order types defined by any set of $n$
  points (not necessarily in general position) is $n^{\Oh(n)}$.

  We use the ideas described in the proof of
  \autoref{lem:enumerate_order_types}. As we did there, we obtain
  points \(q_{1},\dotsc,q_{n}\) so that (i) the order type of
  \(\langle{}q_{1},\dotsc,q_{n}\rangle\) is the same as that of
  \(\langle{}p_{1},\dotsc,p_{n}\rangle\) and (ii) moving $q_i$
  inside a small cell $r_i$ around it does not change orientations
  of triples with non-zero orientations. If required we also shift
  the superimposed grid in such a way that no point $q_i$ lies on
  the boundary of the corresponding region $r_i$. For each point
  $q_i$, we pick a line segment $s_i$ containing $q_i$ such that
  (i) the end-points of $s_i$ lie in $r_i$ and (ii) the $2n$
  end-points of these segments are in general position. Such
  segments can be found by picking the end-points one by one. We
  always pick them inside the appropriate region to satisfy the
  first condition and when we pick the next one, there are only a
  finite number of lines to avoid in order satisfy the second
  condition.  As before, we get that the order type of the ordered
  segment set \(\langle{}s_{1},\dotsc,s_{n}\rangle\) is the same
  as that of the ordered point set
  \(\langle{}p_{1},\dotsc,p_{n}\rangle\). Denote the end-points of
  $s_i$ by $s^1_i$ and $s^2_i$. Then, for any triple of points
  $\langle p_i, p_j, p_k\rangle$ with orientation $+1$ all the
  eight triples of the form $\langle s^u_i, s^v_j,s^w_k\rangle$
  where $u,v,w \in \{1,2\}$ have the orientation $+1$. The same
  holds for triples with orientation $-1$. Furthermore, for any
  triple $\langle p_i, p_j, p_k\rangle$ with orientation $0$, the
  eight triples \emph{do not} have the same orientation. To see
  this, consider the line $\ell$ through $q_i$, $q_j$, and
  $q_k$. Clearly $\ell$ intersects the segments $s_i$, $s_j$, and
  $s_k$. We move $\ell$ parallel to itself until it passes through
  an end-point $x$ of one of the segments, say $s_i$, for the
  first time. (We skip the parallel move if $\ell$ is collinear with
  a segment and just let $x$ be either endpoint. Note that by general
  position of segment endpoints no further endpoints are on $\ell$ in this case.)
  Then, keeping this point \(x\) fixed we rotate
  $\ell$, if required, until it passes through the end-point of
  another segment, say $s_j$, for the first time. (We skip rotation
  if the parallel move already hit two segment endpoints; as they are
  in general position we cannot hit three at the same time.) The line $\ell$
  still intersects $s_k$ and therefore the two end-points of $s_k$
  are on opposite sides of $\ell$ (neither of them can be on
  $\ell$ since we assumed that the set of end-points is in general
  position). Thus $orientation(\langle x, y, s^1_k \rangle) \neq
  orientation(\langle x, y, s^2_k \rangle)$.

  Thus, given the orientations of ordered triples in the set $\{
  s^1_1, s^2_1, \cdots, s^1_n, s^2_n\}$ we can deduce the
  orientations of ordered triples in $\{p_1, \cdots, p_n
  \}$. Hence the order type of $\langle s^1_1, s^2_1, \cdots,
  s^1_n, s^2_n\rangle$ encodes the order type of $\langle p_1,
  \cdots, p_n \rangle$.
\end{proof}

\section{Lower Bound on Kernel Size}\label{sec:reduction}
In this section we prove Theorem~\ref{theorem:plc:sizebound}. The
main component of our proof is a polynomial-time reduction from
\VC to \LPC in which the parameter value is exactly doubled.

\begin{lemma}\label{lemma:reduction}
  There is a polynomial-time reduction from \VC to \LPC which maps
  instances~$(G,k)$ of \VC to equivalent instances~$(\L,2k)$ of
  \LPC. 
\end{lemma}

\begin{proof}
  Given an instance $(G,k)$ of \VC with $n:=|V(G)|$ and
  $m:=|E(G)|$, we construct an equivalent instance \((\L,2k)\) of
  \LPC in two phases.

  In the first phase we construct a graph $G'$ such that $G'$ has
  a vertex cover of size at most $2k$ if and only if $G$ has a
  vertex cover of size at most $k$. To do this, we first make two
  copies \(G_{0},G_{1}\) of the graph \(G\). For a vertex
  \(v\in{}V(G)\), let \(v_{0},v_{1}\) denote its copies in
  \(G_{0}\) and \(G_{1}\), respectively. For each edge
  \(\{u,v\}\in{}E(G)\), we add the two edges
  \(\{u_{0},v_{1}\},\{u_{1},v_{0}\}\) to \(G'\). This completes
  the construction. Note that there are \emph{four} edges in
  \(G'\) which correspond to each edge in \(G\): these consist of
  the two copies of the edge ``within'' \(G_{0}\) and \(G_{1}\),
  and the two edges ``across'' \(G_{0}\) and \(G_{1}\) added in
  the second step. Formally, we set
  \(V(G')=\{v_{0},v_{1}\mid{}v\in V(G)\}\), and
  \(E(G')=\{\{u_i,v_j\}\mid{}i,j\in\{0,1\}\,\wedge\,\{u,v\}\in
  E(G)\}\); see \autoref{fig:phase1}. As we show in
  \autoref{claim:phase1}, $G'$ has a vertex cover of size at most
  $2k$ if and only if $G$ has a vertex cover of size at most $k$.

  \begin{figure}
    \centering
    \includegraphics[scale=0.25]{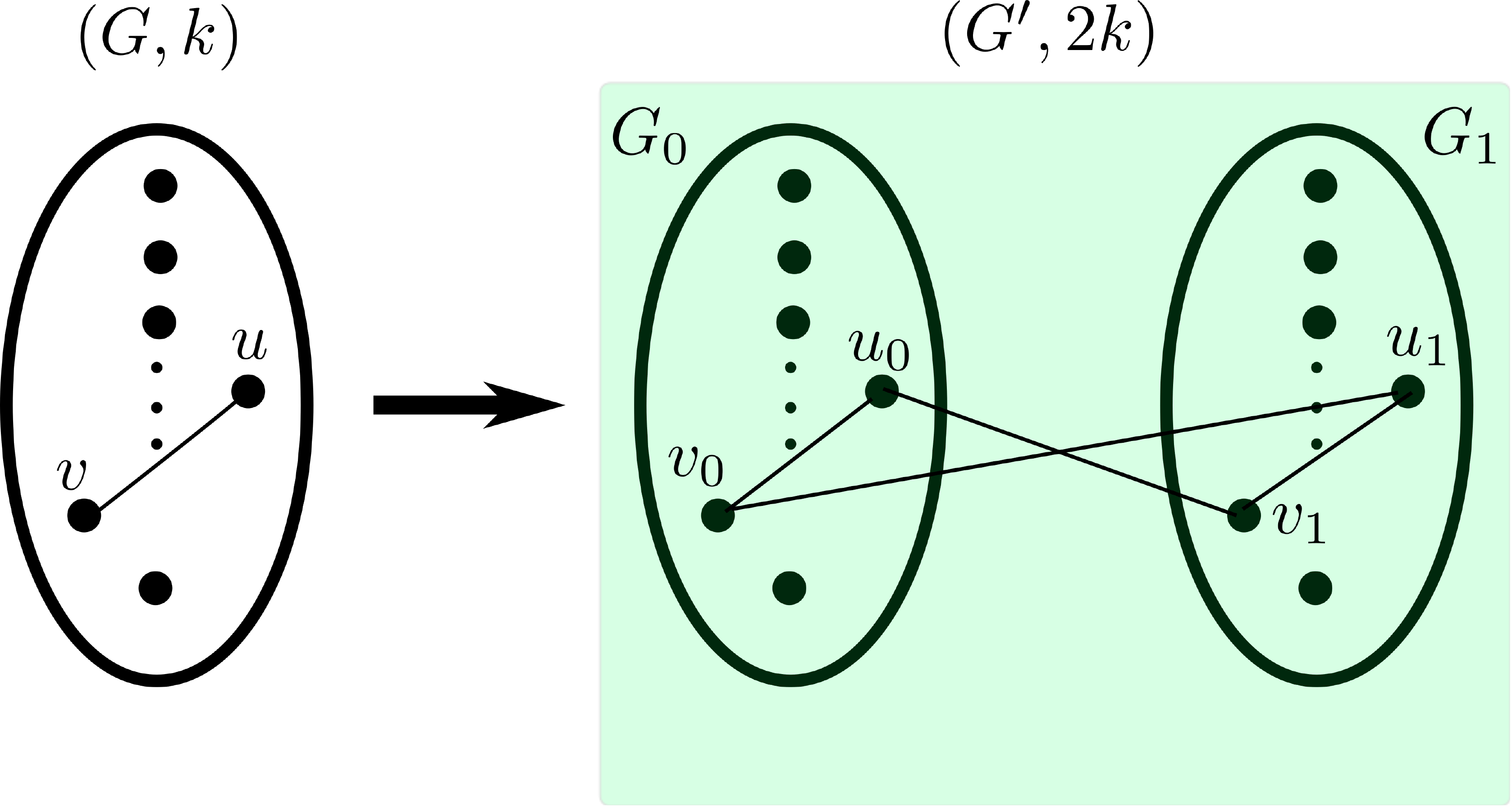}
    \caption{Reduction, phase one.}
    \label{fig:phase1}
  \end{figure}
  
  In the second phase we start with the graph \(G'\) which is the
  output of the first phase, and construct a set \(\L\) of lines
  in the plane. We do this in such a way that there is a set of at
  most \(2k\) points in the plane which cover all the lines in
  \(\L\) if and only if the graph \(G'\) has a vertex cover of
  size at most \(2k\). We start by constructing a set of \(2n\)
  points\footnote{These points will \emph{not} be part of the
    reduced instance.} \(\P=\{p_{1},\ldots,p_{2n}\}\) in the plane
  which has the following properties:
  \begin{enumerate}
  \item The set of points \(\P\) is in general position;
  \item No two lines defined by pairs of points in $\P$ are
    parallel; and,
  \item No three lines defined by pairs of points in $\P$ pass
    through a common point \emph{outside} $\P$.
  \end{enumerate}
  By \autoref{lemma:special_point_set_construction}, we can
  construct the set \(\P\) in time polynomial in \(n\).

  We associate, in an arbitrary fashion, a distinct point
  \(P_{v}\in\P\) with each vertex \(v\in{}V(G')\). We initialize
  \(\L\) to be the empty set. Now for each edge \(\{u,v\}\) of
  \(G'\), we add the line \(L_{uv}=\overline{P_{u}P_{v}}\) to
  \(\L\).  This completes the construction; \((\L,2k)\) is the
  reduced instance of \LPC{}. We now show that this is indeed a
  reduction.

  \medskip\noindent\textbf{Completeness.} Assume that the starting
  instance~$(G,k)$ of \VC is \yes. This
  implies---\autoref{claim:phase1}---that $(G',2k)$ is \yes, and
  it suffices to show that $(\L,2k)$ is \yes for \LPC. Let
  $S\subseteq V(G')$ with $|S|\leq 2k$ be a vertex cover
  of~$G'$. Consider the set~\(\Que=\{P_{v}\mid{}v\in{}S\}\) of at
  most \(2k\) points in the plane. By construction, any line~$L\in
  \L$ is of the form $L=\overline{P_uP_v}$ for some $u,v\in V(G')$
  with $\{u,v\}\in E(G')$. Since \(S\) is a vertex cover of
  \(G'\), we know that \(S\cap\{u,v\}\neq\emptyset\). Hence
  \(\Que\cap\{P_{u},P_{v}\}\neq\emptyset\), and so \(\Que\)
  contains a point which covers the line~$L$. It follows
  that~$\Que$ is a line point cover for~$\L$ of size at most~$2k$
  and that~$(\L,2k)$ is \yes for \LPC.

  \medskip\noindent\textbf{Soundness.} Now suppose $(\L,2k)$ is
  \yes for \LPC. Note that in general, a solution for~$(\L,2k)$
  could contain points which do not belong to the set \(\P\) that
  we used for the construction; these points do not correspond to
  vertices of $G'$. Indeed, any pair of lines in \(\L\) meet at a
  point, but not every pair of edges in~$G'$ share a vertex. We
  show that there exists a solution for $(\L,2k)$ which consists
  entirely of points which correspond to vertices in~$G'$.  We
  start with a \emph{smallest} solution $\Que\,;\,|\Que|\leq2k$ of
  \((\L,2k)\) which has a \emph{minimum number} of points that do
  not correspond to vertices of~$G'$, i.e., with as few
  points~$P\in\Que\setminus\P$ as possible. Let us call all points
  not in~$\P$ \emph{bad points}; points in~$\P$ are \emph{good
    points}.  We define
  \(S(\Que):=\{v\in{}V(G')\mid{}P_{v}\in\Que\}\) as the set of
  vertices of \(G'\) that correspond to good points which are in
  $\Que$.

  Suppose \(\Que\) contains a bad point~$P$. Since \(\Que\) is a
  smallest solution, there is at least one line in \(\L\) which
  (i) is covered by \(P\), and (ii) is not covered by any other
  point in \(\Que\). If there is exactly one such line, say
  $\overline{P_uP_v}$, then we may replace~$P$ by~$P_u$ or~$P_v$
  in \(\Que\) and reduce the number of bad points, a
  contradiction. Now from the third property of the point set
  \(\P\) and from the fact that every line in \(\L\) is defined by
  some pair of points in \(\P\), we get that point~$P$ covers
  \emph{exactly} two lines, say~$L_1,L_2\in \L$, and that these
  are not covered by other points of~$\Que$.

  We now examine the structure around line~$L_1$ more closely. We
  know that~$L_1=\overline{P_{u_i}P_{v_j}}$ for some~$u_i,v_j\in
  V(G')$ with~$i,j\in \{0,1\}$. We assume for the sake of
  convenience that \(i=j=0\), so that
  $L_1=\overline{P_{u_0}P_{v_0}}$; a symmetric argument works when
  \(i=j=1\text{ or }i\neq{}j\). Since~$P$ is the only point
  in~$\Que$ which covers $L_1$, we know that
  $P_{u_0},P_{v_0}\notin \Que$ and hence that $u_0,v_0\notin
  S(\Que)$. We also know from the construction that (i)
  \(\{\overline{P_{u_{i}}P_{v_{j}}}\,;\,i,j\in\{0,1\}\}\subseteq\L\)
  and (ii)
  \(\{\{u_{i},v_{j}\}\,;\,i,j\in\{0,1\}\}\subseteq{}E(G')\). For
  any vertex \(v\in{}V(G')\setminus{}S(\Que)\) the stated minimal
  properties of \(\Que\) imply---see \autoref{claim:phase2}---that
  \(|N_{G'}(v)\setminus{}S(\Que)|\leq1\), and so we get that
  \(N_{G'}(v_{0})\setminus{}S(\Que)=\{u_{0}\}\) and
  \(N_{G'}(u_{0})\setminus{}S(\Que)=\{v_{0}\}\).  It follows that
  \(v_{1}\in{}S(\Que)\) and that \(P_{v_{1}}\) is a good point in
  \(\Que\).

  By construction we have that \(N_{G'}(v_{0})=N_{G'}(v_{1})\),
  and hence that
  \(N_{G'}(v_{1})\setminus{}S(\Que)=\{u_{0}\}\). From the first
  property of the point set \(\P\) we get that the lines in \(\L\)
  which are covered by the point \(P_{v_{1}}\) all correspond to
  edges of \(G'\) which are at \(v_{1}\). It follows that the set
  \(\Que'=(\Que\setminus\{P_{v_{1}}\})\cup\{P_{u_{0}}\}\) covers
  all lines in \(\L\) since it contains the corresponding points $P_w$
  for all neighbors $w$ of $v_1$.  Now (i) \(|\Que'|=|\Que|\), (ii)
  \(P\in\Que'\), (iii) \(\Que'\) has the same number of good
  points as \(\Que\), and (iv) the good points in $\Que'$ cover
  all the lines that were covered by the good points of \(\Que\),
  and in addition they cover
  $L_1=\overline{P_{u_0}P_{v_0}}$. Hence there is at most one line
  in \(\L\)---namely, \(L_{2}\)---which (i) is covered by \(P\),
  and (ii) is not covered by any other point in \(\Que'\). By
  previous arguments this contradicts the minimality of \(\Que'\)
  and hence of \(\Que\).

  It follows that there is a set \(\Que\) of at most \(2k\) points
  which (i) covers all the lines in \(\L\) and (ii) has only good
  points.  We claim that \(S(\Que)\) is a vertex cover of $G'$ of
  size at most $2k$: For any edge~$\{u,v\}\in E(G')$, the
  corresponding line in \(\L\) is covered by a good point
  in~$\Que$, so it is covered by~$P_u$ or~$P_v$. Thus $S(\Que)$
  contains~$u$ or~$v$, and \(|S(\Que)|=|\Que|\leq2k\).  It follows
  that~$(G',2k)$ and hence also~$(G,k)$ are \yes for \VC.

  \medskip\noindent\textbf{Wrap-up.} It is not difficult to see
  that the first phase of the reduction can be done in polynomial
  time, and we have argued that the second phase can also be done
  in polynomial time.  
\end{proof}

We now prove the claims which we use in the proof of
\autoref{lemma:reduction}.

\begin{claim}\label{claim:phase1}
  Let \(G\) be any graph, and let \(G'\) be the graph obtained
  from \(G\) by applying the construction from the first phase of
  the reduction in the proof of \autoref{lemma:reduction}. Then
  \(G\) has a vertex cover of size at most \(k\) if and only if
  \(G'\) has a vertex cover of size at most \(2k\).
\end{claim}
\begin{proof}
  For the forward direction, let \(S\subseteq{}V(G)\) be a vertex
  cover of \(G\), of size at most \(k\). Then
  \(S'=\{v_{0},v_{1}\mid{}v\in{}S\}\) is a vertex cover of \(G'\),
  of size at most \(2k\). Since \(S'\) contains exactly two
  vertices for each vertex of \(S\), we have that
  \(|S'|=2|S|\leq{}2k\). To see that \(S'\) is a vertex cover of
  \(G'\), assume for the sake of contradiction that an edge
  \(\{x,y\}\) of \(G'\) is not covered by \(S'\). There are four
  cases to consider, which reduce by symmetry to the following
  two.
  \begin{enumerate}
  \item Both end points of the edge \(\{x,y\}\) belong to the
    subgraph \(G_{0}\) of \(G\). (A symmetric argument holds when
    both \(x\) and \(y\) belong to the subgraph \(G_{1}\).) Then
    from the construction, there is an edge \(\{u,v\}\in{}E(G)\)
    such that \(x=u_{0},y=v_{0}\). Since neither \(x\) nor \(y\)
    is in \(S'\), it follows from the construction of \(S'\) that
    neither of \(\{u,v\}\) is in \(S\).
  \item Vertices \(x\) and \(y\) belong to two distinct subgraphs
    \(G_{i};i\in\{0,1\}\). Without loss of generality, let
    \(x\in{}V(G_{0}),y\in{}V(G_{1})\). Observe that the
    construction adds such an edge \(\{x,y\}\) to \(G'\) exactly
    when (i) there is an edge \(\{u,v\}\in{}E(G)\), and (ii)
    \((x=u_{0},y=v_{1})\vee(x=v_{0},y=u_{1})\). In either case,
    since neither \(x\) nor \(y\) is in \(S'\), it follows from
    the construction of \(S'\) from \(S\) that neither of
    \(\{u,v\}\) is in \(S\).
  \end{enumerate}
  In each case, \(S\) does not cover the edge \(\{u,v\}\) in
  \(G\), a contradiction.
  
  For the reverse direction, let \(S'\subseteq{}V(G')\) be an
  \emph{inclusion-minimal} vertex cover of \(G'\), of size at most
  \(2k\). If \(G'\) has a vertex cover of size at most \(2k\),
  then it certainly has a minimal vertex cover of size at most \(2k\). We first show
  that for any vertex \(u\in{}V(G)\), either (i) both the copies
  \(u_{0},u_{1}\) are in \(S'\), or (ii) neither copy is in
  \(S'\). Suppose not, and exactly one of these, say \(u_{0}\), is
  in \(S'\). Since \(S'\) is a minimal vertex cover, \(u_{0}\)
  covers at least one edge, and so the neighborhood \(N(u_{0})\)
  of \(u_{0}\) is non-empty. Now, from the construction of \(G'\)
  we get that \(N(u_{0})=N(u_{1})\), and since
  \(u_{1}\notin{}S'\), it follows that \emph{all} vertices in
  \(N(u_{1})=N(u_{0})\) must be in \(S'\). But then
  \(S'\setminus\{u_{0}\}\) is a strictly smaller vertex cover of
  \(G'\), a contradiction.
  
  Thus \(S'\) can be partitioned into pairs of vertices
  \(\{\{u_{0},u_{1}\}\mid{}u\in{}S\subseteq{}V(G)\}\). We claim
  that \(S\) is a vertex cover of graph \(G\), of size at most
  \(k\). The claim on size follows immediately from the
  construction of \(S\). To see that \(S\) is a vertex cover of
  \(G\), observe that if \(\{x,y\}\in{}E(G)\) and
  \(\{x,y\}\cap{}S=\emptyset\), then (i)
  \(\{x_{0},x_{1},y_{0},y_{1}\}\cap{}S'=\emptyset\) by the
  construction of \(S\) and (ii) all the four edges
  \(\{\{x_{i},y_{i}\}\mid{}i\in\{0,1\}\}\) are in \(E(G')\) by the
  construction of \(G'\), which together contradict the assumption
  that \(S'\) is a vertex cover of \(G'\).
\end{proof}

\begin{claim}\label{claim:phase2}
  Let \((\L,2k)\) be an instance of \LPC constructed as in the
  second phase of the reduction in the proof of
  \autoref{lemma:reduction}. Let \(\P,G',P_{v}\) be as in the
  construction, and let \(\Que\) be a smallest set of points which
  (i) cover all the lines in \(\L\) and (ii) has as few points
  $P\in\Que\setminus\P$ as possible. Let
  \(S(\Que):=\{v\in{}V(G')\mid{}P_{v}\in\Que\}\). Then for any
  vertex \(v\in{}V(G')\setminus{}S(\Que)\) we have that
  \(|N_{G'}(v)\setminus{}S(\Que)|\leq1\).
\end{claim}
\begin{proof}
  We reuse the notation and terminology defined in the proof of
  \autoref{lemma:reduction}. Suppose a vertex
  \(v\in{}V(G')\setminus{}S(\Que)\) has two neighbours \(u,w\)
  which are not in \(S(\Que)\). Then the lines
  \(L_{1}=\overline{P_{v}P_{u}}\) and
  \(L_{2}=\overline{P_{v}P_{w}}\) which intersect at the good
  point \(P_{v}\) are covered by two bad points---say
  \(P_{1},P_{2}\), respectively---in \(\Que\). By an argument
  presented in the proof of \autoref{lemma:reduction} we know that
  each of \(P_{1},P_{2}\) covers exactly two lines in \(\L\), and
  that these lines are not covered by any other point in
  \(\Que\). Let \(L_{i},L_{i}'\) be the two lines covered by
  \(P_{i}\) for \(i\in\{1,2\}\), and let \(P'\) be the
  intersection of \(L_{1}'\) and \(L_{2}'\). Thus the two bad
  points \(P_{1},P_{2}\) together cover the subset
  \(\{L_{i},L_{i}'\,\vert\,i\in\{1,2\}\}\) of \(\L\), but so do
  the two points \(P_{v},P'\) of which at most one is a bad
  point. Thus \((\Que\setminus\{P_{1},P_{2}\})\cup\{P_{v},P'\}\)
  is a smallest set of points which (i) cover all the lines in
  \(\L\) and (ii) has strictly fewer bad points than \(\Que\), a
  contradiction.
\end{proof}

Using \autoref{lemma:reduction} we can prove a stronger statement
than \autoref{theorem:plc:sizebound}. 

\begin{theorem}\label{theorem:plc:oraclecostbound}
  Let~$\varepsilon>0$. The \PLC problem admits no oracle
  communication protocol of cost $\Oh(k^{2-\varepsilon})$ for
  deciding instances~$(\P,k)$, unless \caveat.
\end{theorem}

\begin{proof}
  Suppose \PLC admits an oracle communication protocol of cost
  $\Oh(k^{2-\varepsilon})$ that decides any instance
  $(\P,k)$. \autoref{lem:lpc_to_plc} then yields a protocol of the
  same cost for \LPC. Combining this with the reduction from \VC
  to \LPC from \autoref{lemma:reduction} which has a linear
  parameter increase (i.e., $k\mapsto 2k$), we get a protocol of
  cost $\Oh(k^{2-\varepsilon})$ for deciding \VC
  instances~$(G,k)$.  \autoref{theorem:vc_protocol_lower_bound}
  now implies \caveat, as claimed.
\end{proof}

\autoref{theorem:plc:sizebound} is now immediate.
\begin{proof}[Proof of \autoref{theorem:plc:sizebound}]
  If \PLC has a kernel of size $\Oh(k^{2-\varepsilon})$, then the
  polynomially-bounded first player in the oracle communication
  protocol could compute this kernel and send it to the second
  player who, being computationally unbounded, can compute and
  return the correct one-bit answer (\yes or \no). The cost of
  this protocol is $\Oh(k^{2-\varepsilon})$, and by
  \autoref{theorem:plc:oraclecostbound} this implies \caveat.
\end{proof}

\section{Lower bound on the number of points}\label{section:pointlowerbound}

In this section we prove our main result, namely that the
straightforward reduction of \PLC instances~$(\P,k)$ to~$k^2$
points cannot be significantly improved. Recall
that, unlike for \VC, we do not know of an efficient encoding of
\PLC to instances of near linear size in the number of points. In
the case of \VC it is straightforward to encode instances with~$m$
edges using $\Oh(m \log m)$ bits. Together with the known
\(\Oh(k^{2-\varepsilon})\) lower bound on kernel
size~\cite{DellvanMelkebeek2010}, this implies
an~$\Oh(k^{2-\varepsilon})$ lower bound on the number of
\emph{edges} for \VC kernels. Obtaining an efficient encoding of
\PLC to instances of near linear size in the number of points is
an old open problem in computational geometry~\cite{GPS89}. The
following lemma gets around this handicap by providing an oracle
communication protocol of cost near linear in the number of
points.

\begin{lemma}\label{lemma:plc:oracleprotocol}
  There is an oracle communication protocol of cost~$\Oh(n\log n)$
  for deciding instances of \PLC with~$n$ points.
\end{lemma}

\begin{proof}
  We describe the claimed oracle communication protocol for
  deciding \PLC instances. The polynomially-bounded player holding
  the input is called Alice, and the computationally unbounded
  player is called Bob. Recall that by
  \autoref{definition:oracleprotocol} the cost of a protocol is
  the number of bits \emph{sent from Alice to Bob}; in contrast,
  Bob can sent any amount of information to Alice (who, however,
  has only polynomial time in the input size for reading it).
  
  Alice and Bob both use the following scheme to represent order
  types as strings over the alphabet \(\{+1,0,-1\}\).  Recall that
  the order type of an ordered set of \(n\) points
  \(\P=\langle1,\dotsc,n\rangle\) is a certain function
  \(\sigma:\binom{[n]}{3}\mapsto\{+1,0,-1\}\). To form the string
  representing \(\sigma\), we first arrange the set
  \(\binom{[n]}{3}\) in increasing lexicographic order to get a
  list \(\L\). Then we replace each \(x\in\L\) by
  \(\sigma(x)\). This gives us the desired string; we denote it
  the Order Type Representation of the ordered set, or OTR for
  short. Observe that each OTR can be encoded using \(\Oh(n^{3})\)
  bits.

  From \autoref{lem:otub} we know that the number of
  combinatorially distinct order types of \(n\)-point sets is
  $n^{O(n)}$. Since for each order type there are at most $n!$
  other order types combinatorially equivalent to it, the number
  of different order types of \(n\)-point sets is at most
  $n!\cdot{}n^{O(n)}=n^{O(n)}$. This indicates that the pertinent
  information which Alice holds is $\Oh(n\log n)$ bits.  But we do
  not know of a polynomial-time procedure which encodes this
  information into these many bits, and so Alice cannot just
  employ such an encoding and send the result to Bob for him to
  solve the instance.  We use the following protocol to get around
  this problem.
  
  \begin{enumerate}
    \item Alice sends the value \(n\) of the number of points in
      the input set to Bob in binary encoding. 
    \item Alice fixes an arbitrary ordering of the input point
      set. She then computes the OTR of this ordered set.
    \item Bob generates a list of all $n^{O(n)}$ possible order
      types; by \autoref{lem:enumerate_order_types} this is a
      computable task. He then computes the OTRs of these order
      types and sorts them in lexicographically increasing order.
    \item Alice and Bob now engage in a conversation where Bob
      uses binary search on the sorted list to locate the OTR
      which Alice holds. Bob sends the median OTR \(M\) in his
      list to Alice. Alice replies, in two bits, whether the OTR
      she holds is smaller, equal to, or larger than \(M\) in
      lexicographic order. If the answer is not ``equal'', Bob
      prunes his list accordingly, throwing out all OTRs which
      cannot be the one held by Alice. By repeating this procedure
      $\Oh(\log (n^{O(n)}))=\Oh(n\log n)$ times, Bob is left with
      a single OTR \(S\) which is identical to the one held by Alice.  
    \item Bob now computes the size of a smallest point-line cover
      of any point set which has the order type \(S\), and sends
      this number to Alice. Alice compares this number with the
      input \(k\) and answers \yes or \no accordingly. 
  \end{enumerate}
  
  It is not difficult to see that Alice can do her part of this
  procedure in polynomial time, and that all tasks which Bob has
  to do are computable. The total cost of the protocol is
  \(\log{}n+\Oh(n\log{}n)=\Oh(n\log{}n)\), as claimed. 
\end{proof}

This lemma and the kernel on \(\Oh(k^{2})\) points together imply
an oracle communication protocol for \PLC that matches the lower
bound from \autoref{theorem:plc:oraclecostbound} up to~$k^{o(1)}$
factors. This suggests that a kernelization or compression to bit
size~$\Oh(k^{2+o(1)})$ may be possible; at least it is impossible
to get better lower bounds via oracle communication protocols.

\begin{corollary}
  The \PLC problem has an oracle communication protocol of cost
  $\Oh(k^2\log k)$ for deciding instances $(\P,k)$.
\end{corollary}

Now, using the protocol from \autoref{lemma:plc:oracleprotocol} in
place of an efficient encoding, it is straightforward to complete
the claimed lower bound of~$\Oh(k^{2-\varepsilon})$ on the number
of points in a kernel.

\begin{proof}[Proof of Theorem~\ref{theorem:plc:pointlowerbound}]
  By \autoref{lemma:plc:oracleprotocol}, such a kernelization
  would directly give an oracle communication protocol for \PLC of
  cost~$\Oh(k^{2-\varepsilon'})$: Given an instance~$(\P,k)$,
  Alice applies the (polynomial-time) kernelization that generates
  an equivalent instance with~$\Oh(k^{2-\varepsilon})$
  points. Then she proceeds by using the protocol from the proof
  of \autoref{lemma:plc:oracleprotocol}.

  As we already showed in \autoref{theorem:plc:oraclecostbound},
  there is no~$\Oh(k^{2-\varepsilon'})$ protocol for \PLC for
  any~$\varepsilon'>0$, unless \caveat. This completes the proof.
\end{proof}

\section{Conclusion}\label{sec:conclusion}

We took up the question of whether the known simple reduction of
\PLC instances $(\P,k)$ to equivalent instances with $k^2$ points
can be significantly improved. This has been posed as an open
problem by Lokshtanov in his PhD
Thesis~\cite{LokshtanovThesis2009} and also at the open problem
sessions of various meetings of the Parameterized Algorithms
community. Our main result, \autoref{theorem:plc:pointlowerbound},
answers Lokshtanov's question in the negative.

Along the way, we proved that no polynomial-time reduction to
\emph{size} $\Oh(k^{2-\varepsilon})$ bits is possible either. Now,
starting with the reduction to $k^2$ points, one can encode the
order type of the reduced instance using \emph{lambda
  matrices}~\cite{GoodmanPollack1991} into $\Oh(k^4\log k)$
bits. Let us recall however, that our lower bound (like all lower
bounds via the same framework~\cite{DellvanMelkebeek2010}) holds
also for oracle communication protocols. Since we devise a
protocol of cost~$\Oh(k^2\log k)$, our lower bound seems to be the
best possible with these methods. We pose as an open problem
whether the gap between the upper and lower bound in the total
size of a reduced instance can be closed; we expect that the
``correct bound'' is~$\tilde\Oh(k^2)$ bits.

We briefly mention a variant of the \textsc{Set Cover} problem
where sets are restricted to have pairwise intersections of
cardinality at most one ($1$-ISC)~\cite{KumarAryaRamesh2000}, of
which \PLC is a special case. It is not hard to see that the
reduction to $k^2$ points carries over directly to a reduction of
the ground set of such a $1$-ISC instance to $k^2$ elements if we
are looking for a set cover of size most $k$.  This also gives a
polynomial kernelization of the problem of size~$\Oh(k^5\log k)$
bits since there are at most~$\binom{k^2}{2}$ sets (each pair of
ground set elements is in at most one of them), each containing at
most~$k$ elements whose identity can be encoded in~$\Oh(\log k)$
bits. (A more careful analysis yields an upper bound of
\(\Oh(k^4\log k)\) bits.)  The lower bound for the cost of oracle
communication protocols for PLC carries over to this problem as
well, and so $1$-ISC has no such protocol or kernelization of
cost/size $\Oh(k^{2-\varepsilon})$ for any $\varepsilon>0$ unless
\caveat. However, we do \emph{not} know of a protocol for deciding
$1$-ISC instances with ground set of $n$ elements at cost
$\Oh(n^{1+o(1)})$, and so we do not have a bound for the size of
the ground set. Thus, unlike for PLC, we also have no protocol of
cost $\Oh(k^2\log k)$ for \(1\)-ISC to (essentially) match the
lower bound for protocols. Note that the lower bound on the number
of points does not transfer from PLC to (the ground set size of)
$1$-ISC since using a reduction for $1$-ISC on a PLC instance does
not necessarily result in a PLC instance. We ask whether the gap
between upper and lower bounds for kernels for $1$-ISC can be
closed, and whether, perhaps similar to our protocol for PLC, a
tight lower bound for the ground set elements can be proven.
\newpage
\bibliographystyle{plain}
\bibliography{pointslines}

\begin{thebibliography}{10}

\bibitem{Abu-Khzam10}
Faisal~N. Abu-Khzam.
\newblock A kernelization algorithm for d-hitting set.
\newblock {\em J. Comput. Syst. Sci.}, 76(7):524--531, 2010.

\bibitem{BodlaenderDFH09}
Hans~L. Bodlaender, Rodney~G. Downey, Michael~R. Fellows, and Danny Hermelin.
\newblock On problems without polynomial kernels.
\newblock {\em J. Comput. Syst. Sci.}, 75(8):423--434, 2009.

\bibitem{BrodenHammarNilsson2001}
Bj{\"o}rn Brod{\'e}n, Mikael Hammar, and Bengt~J. Nilsson.
\newblock Guarding lines and 2-link polygons is {APX}-hard.
\newblock In {\em Proceedings of the 13th Canadian Conference on Computational
  Geometry, University of Waterloo, Ontario, Canada, August 13-15, 2001}, pages
  45--48, 2001.

\bibitem{ChenKJ01}
Jianer Chen, Iyad~A. Kanj, and Weijia Jia.
\newblock Vertex cover: Further observations and further improvements.
\newblock {\em J. Algorithms}, 41(2):280--301, 2001.

\bibitem{CyganGH13_arxiv}
Marek Cygan, Fabrizio Grandoni, and Danny Hermelin.
\newblock Tight kernel bounds for problems on graphs with small degeneracy.
\newblock {\em CoRR}, abs/1305.4914, 2013.

\bibitem{M4}
Mark de~Berg, Otfried Cheong, Marc van Kreveld, and Mark Overmars.
\newblock {\em Computational Geometry}.
\newblock Springer-Verlag, 3rd revised edition, 2008.

\bibitem{DellM12}
Holger Dell and D{\'a}niel Marx.
\newblock Kernelization of packing problems.
\newblock In Rabani \cite{DBLP:conf/soda/2012}, pages 68--81.

\bibitem{DellvanMelkebeek2010}
{H}olger {D}ell and {D}ieter van {M}elkebeek.
\newblock {S}atisfiability {A}llows {N}o {N}ontrivial {S}parsification {U}nless
  {T}he {P}olynomial-{T}ime {H}ierarchy {C}ollapses.
\newblock In {\em Proceedings of the 42nd ACM Symposium on Theory of Computing,
  STOC 2010, Cambridge, Massachusetts, USA, 5-8 June 2010}, pages 251--260.
  ACM, 2010.

\bibitem{Diestel2005}
{R}einhard {D}iestel.
\newblock {\em {G}raph {T}heory}.
\newblock Springer-Verlag, Heidelberg, {T}hird edition, 2005.

\bibitem{DowneyFellows1999}
{R}od~{G}. {D}owney and {M}ichael~{R}. {F}ellows.
\newblock {\em {P}arameterized {C}omplexity}.
\newblock Springer-Verlag, New York, 1999.

\bibitem{Estivill-CastroHeednacramSuraweera2009}
Vladimir Estivill-Castro, Apichat Heednacram, and Francis Suraweera.
\newblock Reduction rules deliver efficient fpt-algorithms for covering points
  with lines.
\newblock {\em ACM Journal of Experimental Algorithmics}, 14, 2009.

\bibitem{FlumGrohe2006}
{J}{\"o}rg {F}lum and {M}artin {G}rohe.
\newblock {\em {P}arameterized {C}omplexity {T}heory}.
\newblock Springer-Verlag, 2006.

\bibitem{FortnowS11}
Lance Fortnow and Rahul Santhanam.
\newblock Infeasibility of instance compression and succinct {PCP}s for {NP}.
\newblock {\em J. Comput. Syst. Sci.}, 77(1):91--106, 2011.

\bibitem{GoodmanPollack1991}
Jacob~E. Goodman and Richard Pollack.
\newblock The complexity of point configurations.
\newblock {\em Discrete Applied Mathematics}, 31:167--180, 1991.

\bibitem{GPS89}
Jacob~E. Goodman, Richard Pollack, and Bernd Sturmfels.
\newblock Coordinate representation of order types requires exponential
  storage.
\newblock In David~S. Johnson, editor, {\em STOC}, pages 405--410. ACM, 1989.

\bibitem{GrantsonLevcopoulos2006}
Magdalene Grantson and Christos Levcopoulos.
\newblock Covering a set of points with a minimum number of lines.
\newblock In Tiziana Calamoneri, Irene Finocchi, and Giuseppe~F. Italiano,
  editors, {\em Algorithms and Complexity, 6th Italian Conference, CIAC 2006,
  Rome, Italy, May 29-31, 2006, Proceedings}, volume 3998 of {\em Lecture Notes
  in Computer Science}, pages 6--17. Springer, 2006.

\bibitem{GrigorevVorobjov1988}
D.~Yu. Grigor'ev and N.N. Vorobjov~Jr.
\newblock Solving systems of polynomial inequalities in subexponential time.
\newblock {\em Journal of Symbolic Computation}, 5(1–2):37 -- 64, 1988.

\bibitem{HassinMegiddo1991}
Refael Hassin and Nimrod Megiddo.
\newblock {A}pproximation algorithms for hitting objects with straight lines.
\newblock {\em Discrete Applied Mathematics}, 20:29--42, 1991.

\bibitem{Heednacram2010}
{A}pichat {H}eednacram.
\newblock {\em The {NP}-Hardness of {C}overing {P}oints with {L}ines, {P}aths
  and {T}ours and their {T}ractability with {FPT}-{A}lgorithms}.
\newblock Phd thesis, Griffith University, Australia, 2010.

\bibitem{HermelinW12}
Danny Hermelin and Xi~Wu.
\newblock Weak compositions and their applications to polynomial lower bounds
  for kernelization.
\newblock In Rabani \cite{DBLP:conf/soda/2012}, pages 104--113.

\bibitem{Karp1972}
{R}ichard~{M}. {K}arp.
\newblock {R}educibility {A}mong {C}ombinatorial {P}roblems.
\newblock In {\em Proceedings of a symposium on the Complexity of Computer
  Computations, held March 20-22, 1972, at the IBM Thomas J. Watson Research
  Center, Yorktown Heights, New York}, The IBM Research Symposia Series, pages
  85--103. Plenum Press, New York, 1972.

\bibitem{KumarAryaRamesh2000}
V.~S.~Anil Kumar, Sunil Arya, and H.~Ramesh.
\newblock Hardness of set cover with intersection 1.
\newblock In Ugo Montanari, Jos{\'e} D.~P. Rolim, and Emo Welzl, editors, {\em
  Automata, Languages and Programming, 27th International Colloquium, ICALP
  2000, Geneva, Switzerland, July 9-15, 2000, Proceedings}, volume 1853 of {\em
  Lecture Notes in Computer Science}, pages 624--635. Springer, 2000.

\bibitem{Lampis2011}
Michael Lampis.
\newblock A kernel of order $2k-c\log{} k$ for vertex cover.
\newblock {\em Information Processing Letters}, 111(23-24):1089--1091, 2011.

\bibitem{LangermanMorin2005}
Stefan Langerman and Pat Morin.
\newblock Covering things with things.
\newblock {\em Discrete {\&} Computational Geometry}, 33(4):717--729, 2005.

\bibitem{LokshtanovThesis2009}
{D}aniel {L}okshtanov.
\newblock {\em {N}ew {M}ethods in {P}arameterized {A}lgorithms and
  {C}omplexity}.
\newblock PhD thesis, University of Bergen, Norway, 2009.

\bibitem{LokshtanovMS12}
Daniel Lokshtanov, Neeldhara Misra, and Saket Saurabh.
\newblock Kernelization - preprocessing with a guarantee.
\newblock In Hans~L. Bodlaender, Rod Downey, Fedor~V. Fomin, and D{\'a}niel
  Marx, editors, {\em The Multivariate Algorithmic Revolution and Beyond},
  volume 7370 of {\em Lecture Notes in Computer Science}, pages 129--161.
  Springer, 2012.

\bibitem{MegiddoTamir1982}
Nimrod Megiddo and Arie Tamir.
\newblock {O}n the {C}omplexity of {L}ocating {L}inear {F}acilities in the
  {P}lane.
\newblock {\em Operations Research Letters}, 1(5):194--197, 1982.

\bibitem{Niedermeier2006}
{R}olf {N}iedermeier.
\newblock {\em {I}nvitation to {F}ixed-{P}arameter {A}lgorithms}.
\newblock Oxford University Press, 2006.

\bibitem{DBLP:conf/soda/2012}
Yuval Rabani, editor.
\newblock {\em Proceedings of the Twenty-Third Annual ACM-SIAM Symposium on
  Discrete Algorithms, SODA 2012, Kyoto, Japan, January 17-19, 2012}. SIAM,
  2012.

\bibitem{WangLiChen2010}
Jianxin Wang, Wenjun Li, and Jianer Chen.
\newblock A parameterized algorithm for the hyperplane-cover problem.
\newblock {\em Theoretical Computer Science}, 411(44-46):4005--4009, 2010.

\bibitem{Yap1983}
Chee-Keng Yap.
\newblock Some consequences of non-uniform conditions on uniform classes.
\newblock {\em Theoretical Computer Science}, 26:287--300, 1983.

\end{thebibliography}

\end{document}